\documentclass[11pt,letterpaper]{article}
\newcommand{\PAPER}[1]{#1}
\newcommand{\LIPICS}[1]{}

\PAPER{
\usepackage[letterpaper]{geometry}
\usepackage{fullpage,amsthm,amssymb,graphicx,hyperref}

\title{Tree Drawings Revisited}
\author{Timothy M. Chan\thanks{Department of Computer Science, University of Illinois at Urbana-Champaign (tmc@illinois.edu).}}

\newtheorem{theorem}{Theorem}[section]
\newtheorem{lemma}[theorem]{Lemma}
\newtheorem{corollary}[theorem]{Corollary}
\newtheorem{obs}[theorem]{Observation}
\newtheorem{fact}[theorem]{Fact}
}

\LIPICS{
\usepackage{microtype}

\title{Tree Drawings Revisited}
\author{Timothy M. Chan}{Department of Computer Science, University of Illinois at Urbana-Champaign, USA}{tmc@illinois.edu}{}{}

\authorrunning{T.\,M. Chan}
\Copyright{Timothy M. Chan}

\subjclass{
, 
, 
}


\relatedversion{A full version of this paper is available at
\url{XXX}.}

\EventEditors{Bettina Speckmann and Csaba D. T{\'o}th}
 \EventNoEds{2}
 \EventLongTitle{34th International Symposium on Computational Geometry (SoCG 2018)}
 \EventShortTitle{SoCG 2018}
 \EventAcronym{SoCG}
 \EventYear{2018}
 \EventDate{June 11--14, 2018}
 \EventLocation{Budapest, Hungary}
 \EventLogo{socg-logo} 
 \SeriesVolume{99}
 \ArticleNo{23} %

\theoremstyle{plain}
\newtheorem{obs}[theorem]{Observation}
\newtheorem{fact}[theorem]{Fact}

}

\newcommand{\fig}[3]{\begin{figure}%
  \begin{center}\includegraphics[scale=#2]{#1.pdf}\end{center}%
  \vspace{-2ex}\caption{#3}\label{#1}\end{figure}}

\newcommand{\up}[1]{\left\lceil #1\right\rceil}
\newcommand{\down}[1]{\left\lfloor #1\right\rfloor}
\newcommand{\eps}{\varepsilon}
\newcommand{\R}{\mathbb{R}}

\newcommand{\IGNORE}[1]{}

\LIPICS{\nolinenumbers}
\begin{document}
\maketitle

\begin{abstract}
We make progress on a number of open problems concerning the area requirement for drawing trees on a grid.  We prove that
\begin{enumerate}
\item every tree of size $n$ (with arbitrarily large degree) has a straight-line drawing with area $n2^{O(\sqrt{\log\log n\log\log\log n})}$, improving the longstanding $O(n\log n)$ bound;
\item every tree of size $n$ (with arbitrarily large degree) has a straight-line upward drawing with area $n\sqrt{\log n}(\log\log n)^{O(1)}$, improving the longstanding $O(n\log n)$ bound;
\item every binary tree of size $n$ has a straight-line orthogonal drawing with area $n2^{O(\log^*n)}$, improving the previous $O(n\log\log n)$ bound by Shin, Kim, and Chwa (1996) and Chan, Goodrich, Kosaraju, and Tamassia (1996);
\item every binary tree of size $n$ has a straight-line order-preserving drawing with area $n2^{O(\log^*n)}$, improving the previous $O(n\log\log n)$ bound by Garg and Rusu (2003);
\item every binary tree of size $n$ has a straight-line orthogonal order-preserving drawing with area $n2^{O(\sqrt{\log n})}$, improving the $O(n^{3/2})$ previous bound by Frati (2007).
\end{enumerate}
\end{abstract}

\section{Introduction}

Drawing graphs with small area 
has been a subject of intense study in combinatorial
and computational geometry for more than two decades \cite{DiBOOK,DiFraSURV}. 
The goal is to determine worst-case
bounds on the area needed to draw
any $n$-vertex graph in a given class, subject to
certain drawing criteria,
where vertices
are mapped to points on an integer grid $\{1,\ldots,W\}\times\{1,\ldots,H\}$, and the \emph{area}
of the drawing is defined to be the width $W$ times the height $H$.
All drawings in this paper are required to be {\em planar\/}, where
edge crossings are not allowed.
All our results will be about {\em straight-line\/} drawings,
where edges are drawn as straight line segments,
although {\em poly-line\/} drawings that allow bends along the edges
have also received considerable attention.

It is well known~\cite{FrPaPo,Sch} that every planar graph of size $n$
has a straight-line drawing with area $O(n^2)$ (with width and height $O(n)$), and this bound
is asymptotically tight in the worst case.  Much 
research is devoted to understanding which subclasses of planar 
graphs admit subquadratic-area drawings, and obtaining tight area bounds
for such classes.

\PAPER{\subsection{Drawing arbitrary trees}}
\LIPICS{\subparagraph{Drawing arbitrary trees.}}

Among the simplest is the class of all trees.
As hierarchical structures occur naturally in many areas
(from VLSI design to phylogeny),
visualization of trees is of particular interest.
Although there have been numerous papers on tree drawings
(e.g., \cite{Bac,BieJGAA17,BieCCCG17,SODA99,GD96,Cre,CrePen,Fra07,FraSODA17,GarSoCG93,
GarRus03,GarRusICCSA,GarRus04,LeeTHESIS,Lei,ReiTil,ShKiCh,ShiIPL,ShiTHESIS,Tre,Val}),
the most basic question of determining the worst-case area 
needed to draw arbitrary trees, without any additional criteria other
than being planar and straight-line, is surprisingly still open.

An $O(n\log n)$ area upper bound is folklore
and can be obtained by a straightforward recursive algorithm, as described in Figure~\ref{fig_standard}, which we will refer to
as the \emph{standard\/} algorithm
(the earliest reference was perhaps Shiloach's 1976 thesis~\cite[page~94]{ShiTHESIS};
see also Crescenzi, Di Battista, and Piperno~\cite{Cre} for the same algorithm for binary trees).  The algorithm gives linear width 
and logarithmic height.  An analogous algorithm, with $x$
and $y$ coordinates swapped, gives logarithmic width and
linear height. 


However, no single improvement to the
$O(n\log n)$ bound has been found for general trees.
No improvement is known even if drawings are relaxed to
be poly-line!

In an early
SoCG'93 paper by Garg, Goodrich, and Tamassia~\cite{GarSoCG93},
it was shown that
linear area is attainable for poly-line drawings of trees 
with degree bounded by $O(n^{1-\eps})$ for any constant $\eps>0$.
Later, Garg and Rusu~\cite{GarRus04,GarRusICCSA} obtained a similar result
for straight-line drawings for degree up to $O(n^{1/2-\eps})$.\footnote{
It is not clear to this author if their analysis assumed a much stronger property,
that every subtree of size $m$ has degree at most $O(m^{1/2-\eps})$.
}
These approaches do not give good bounds when the maximum degree is linear.

\PAPER{
\fig{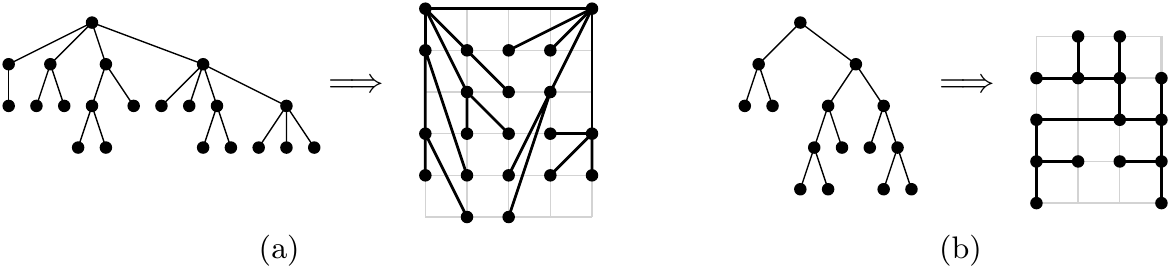}{1.15}{Examples of tree drawings: (a)~a straight-line upward drawing with width 5 and height 6, and (b)~a straight-line
orthogonal order-preserving drawing with width 4 and height 5.}
}

\fig{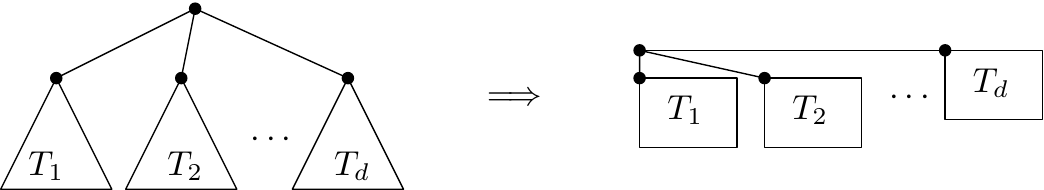}{\PAPER{1}\LIPICS{0.8}}{The ``standard'' algorithm to
produce a straight-line upward drawing of any tree of size $n$,
with width at most $n$ and height at most $\up{\log n}$:
reorder the subtrees so that $T_d$ is the largest, then
recursively draw $T_1,\ldots,T_d$.}

To understand why unbounded degree can pose extra challenges,
consider the extreme case when the tree is a star of size $n$,
and we want to draw it on an $O(\sqrt{n})\times O(\sqrt{n})$ grid.
A solution is not difficult if we use the 
fact that relatively prime pairs are abundant, but
most tree drawing algorithms use geometric divide-and-conquer
strategies that do not seem compatible with such number-theoretic
ideas.

\PAPER{\paragraph{New results.}}
\LIPICS{\subparagraph{New results.}}
Our first main result is the first $o(n\log n)$ area upper bound
for straight-line drawings of arbitrary trees: the bound is $n2^{O(\sqrt{\log\log n\log\log\log n})}$, which in particular
is better than $O(n\log^\eps n)$ for any constant
$\eps>0$.  

Even to those who care less
about refining logarithmic factors, our method has
one notable advantage: it can give drawings achieving
a full range of width--height tradeoffs (in other words,
a full range of aspect ratios).  For example, we can simultaneously obtain
width and height $\sqrt{n}2^{O(\sqrt{\log n\log\log n})}$.
Although the extra factor is now superpolylogarithmic,
the result is still new.
In contrast, the standard algorithm (Figure~\ref{fig_standard}) produces only
narrow drawings, whereas the previous approaches of
Garg et al.~\cite{GarSoCG93,GarRus04} provided width--height
tradeoffs but inherently cannot
give near $\sqrt{n}$ perimeter if degree exceeds $\sqrt{n}$. 

For rooted trees, it is natural to consider
\emph{upward\/} drawings, where the $y$-coordinate of
each node is greater than or equal to the $y$-coordinate
of each child\PAPER{  (see Figure~\ref{fig_ex}(a))}.  The drawing obtained
by the standard
algorithm is upward.
We obtain the first $o(n\log n)$ area bound
for straight-line upward drawings of arbitrary trees as well:
the bound is near $O(n\sqrt{\log n})$, ignoring small $\log\log$ factors.  (See Table~1.)

These results represent significant progress towards
Open Problems 5, 6, 17, and 18 listed in Di Battista and Frati's
recent survey~\cite{DiFraSURV}.

We will describe the near-$O(n\sqrt{\log n})$ upward algorithm first,
in Section~\ref{sec:general}, which prepares us for the more involved
$n2^{O(\sqrt{\log\log n\log\log\log n})}$ non-upward algorithm
in Section~\ref{sec:general2}.


\begin{table}
\begin{tabular}{l|l|l}
  & non-order-preserving & order-preserving\\[.5ex]\hline&&\\[-1.5ex]
\begin{tabular}{l}non-\\upward\end{tabular} & 
\begin{tabular}{ll}$O(n\log n)$
& by standard alg'm\\ 
$O(nc^{\sqrt{\log\log n\log\log\log n}})$ & {\bf new} \end{tabular} 
& 
\begin{tabular}{l}$O(n\log n)$\\ by Garg--Rusu'03~\cite{GarRus03}\end{tabular}
\\[3.5ex]\hline&&\\[-1.5ex]
\begin{tabular}{l}upward\end{tabular} & 
\begin{tabular}{ll}$O(n\log n)$ & by standard alg'm\\ 
$O(n\sqrt{\log n}\log^c\log n)$ & {\bf new}\end{tabular} 
& 
\begin{tabular}{ll}$O(nc^{\sqrt{\log n}})\!\!$ & by Chan'99~\cite{SODA99}$\!\!$\end{tabular}
\\[3.5ex]\hline&&\\[-1.5ex]
\begin{tabular}{l}strictly\\upward\end{tabular} & 
\begin{tabular}{ll} $\Theta(n\log n)$ & by standard alg'm~\cite{Cre}  \end{tabular}
&  
\begin{tabular}{ll}$O(nc^{\sqrt{\log n}})\!\!$ & by Chan'99~\cite{SODA99}$\!\!$\end{tabular}
\end{tabular}
\vspace{1ex}
\caption{Worst-case area bounds for straight-line drawings
of \emph{arbitrary trees}.
(In all tables, $c$ denotes some constant, and $\Theta$
denotes
tight results that have matching lower bounds.)}
\end{table}

\PAPER{\subsection{Drawing binary trees}}
\LIPICS{\subparagraph{Drawing binary trees.}}

\begin{table}
\begin{tabular}{l|l|l}
  & non-order-preserving & order-preserving\\[.5ex]\hline&&\\[-1ex]
\begin{tabular}{l}non-\\upward\end{tabular} & 
\begin{tabular}{l}$\Theta(n)$\\ by Garg--Rusu'04~\cite{GarRus04}
\end{tabular}
& 
\begin{tabular}{ll}$O(n\log\log n)$ & by Garg--Rusu'03~\cite{GarRus03}\\
$O(nc^{\log^*n})$ & {\bf new}\end{tabular}
\\[3ex]\hline&&\\[-1ex]
\begin{tabular}{l}upward\end{tabular} & 
\begin{tabular}{l}$O(n\log\log n)$\\by Shin--Kim--Chwa'96~\cite{ShKiCh}\end{tabular} 
& 
\begin{tabular}{ll}
$O(n^{1.48})$ & by Chan'99~\cite{SODA99}\\
$O(nc^{\sqrt{\log n}})$ & by Chan'99~\cite{SODA99}\\
$O(n\log n)$ & by Garg--Rusu'03~\cite{GarRus03}\end{tabular}
\\[4.5ex]\hline&&\\[-1ex]
\begin{tabular}{l}strictly\\ upward \end{tabular} & 
\begin{tabular}{l}$\Theta(n\log n)$\\ by 
standard alg'm~\cite{Cre}\end{tabular} &
\begin{tabular}{ll}
$O(n^{1.48})$ & by Chan'99~\cite{SODA99}\\
$O(nc^{\sqrt{\log n}})$ & by Chan'99~\cite{SODA99}\\
$\Theta(n\log n)$ & by Garg--Rusu'03~\cite{GarRus03}\end{tabular}
\end{tabular}
\vspace{1.5ex}
\caption{Worst-case area bounds for straight-line drawings of \emph{binary trees}.
}
\end{table}

Next we turn to drawings of binary trees, where there has
been a large body of existing work, due to the many combinations
of aesthetic criteria that may be imposed.  We may
consider
\begin{itemize}
\item {\em upward\/} drawings, as defined earlier;
\item {\em strictly upward\/} drawings, where 
the $y$-coordinate of
each node is strictly greater the $y$-coordinate
of each child;
\item {\em order-preserving\/} drawings, where
the order of children of each node $v$ is preserved, i.e., 
the parent, the left child, and the right child of $v$ appear
in counterclockwise order around~$v$;
\item {\em orthogonal\/} drawings, where
all edges are drawn with horizontal or vertical line segments\PAPER{ 
(see Figure~\ref{fig_ex}(b))}.
\end{itemize}

Tables 2--3 summarize the dizzying array of known 
results on straight-line drawings.  (To keep the table size down,
we omit numerous other results on
poly-line drawings, and
on special subclasses of balanced trees.
See Di Battista and Frati's survey~\cite{DiFraSURV} for more.)

\PAPER{\paragraph{New results.}}
\LIPICS{\subparagraph{New results.}}
In this paper, we concentrate on two of the previous $O(n\log\log n)$
entries in the table.
In 1996, Shin, Kim, and Chwa~\cite{ShKiCh} and Chan et al.~\cite{GD96}
independently obtained $O(n\log\log n)$-area algorithms for
straight-line orthogonal drawings of binary trees;
a few years later, Garg and Rusu~\cite{GarRus03}
adapted their technique to obtain similar results for straight-line
(non-orthogonal) order-preserving drawings.
We improve the area bound for both types of drawings 
to \emph{almost\/} linear: $n2^{O(\log^*n)}$, where $\log^*$ denotes the
iterated logarithm.  (We can also obtain width--height tradeoffs for these drawings.)

\begin{table}
\begin{tabular}{l|l|l}
  & non-order-preserving & order-preserving\\[.5ex]\hline&&\\[-1ex]
\begin{tabular}{l}non-\\upward\end{tabular} & 
\begin{tabular}{ll}$O(n\log\log n)$ & by Chan--Goodrich--\\
& Kosaraju--Tamassia'96~\cite{GD96}\\
& \& Shin--Kim--Chwa'96~\cite{ShKiCh}\\
$O(nc^{\log^* n})$ & {\bf new}
\end{tabular}
& 
\begin{tabular}{ll}$O(n^{3/2})$ & Frati'07~\cite{Fra07}\\
$O(nc^{\sqrt{\log n}})$ & {\bf new}
\end{tabular}
\\[5.5ex]\hline&&\\[-1ex]
\begin{tabular}{l}upward\end{tabular} & 
\begin{tabular}{ll}$\Theta(n\log n)$ & by standard alg'm~\cite{Cre}\end{tabular} 
& \begin{tabular}{l}$\Theta(n^2)$\end{tabular}
\end{tabular}
\vspace{1.5ex}
\caption{Worst-case area bounds for straight-line \emph{orthogonal\/} drawings of \emph{binary trees}.  (Strictly upward drawings are not
possible here.)} 
\end{table}

Although improving $\log\log n$ to iterated logarithm may
not come as a total surprise,
the problem for straight-line
orthogonal drawings has resisted attack for 20 years.
(Besides, improvement should not be taken for granted,
since there is at least one class of drawings
for which $\Theta(n\log\log n)$ turns out to be tight: poly-line upward orthogonal drawings of binary trees~\cite{GarSoCG93}.)

We have additionally one more result on straight-line orthogonal
order-preserving drawings of binary trees: in 2007,
Frati~\cite{Fra07} presented an $O(n^{3/2})$-area algorithm.  We
improve the bound to $n2^{O(\sqrt{\log n})}$, which in particular
is better than $O(n^{1+\eps})$ for any constant $\eps>0$.

These results represent significant progress towards
Open Problems 9, 12, and 14 listed in Di Battista and Frati's
survey~\cite{DiFraSURV}.


(The author has obtained still more new results, on a special class of
so-called {\em LR drawings\/} of binary trees~\cite{SODA99,FraSODA17},
making progress on Open Problem~10 in the survey,
which will be reported later elsewhere.)

We will describe the $n2^{O(\log^*n)}$ algorithm for
orthogonal drawings first,
in Section~\ref{sec:logstar}; the algorithm
for non-orthogonal order-preserving drawings is similar,
\PAPER{as noted in Section~\ref{sec:logstar2}}\LIPICS{and is described in the full paper}.
The $n2^{O(\sqrt{\log n})}$ algorithm for orthogonal order-preserving drawings
\PAPER{is presented in Section~\ref{sec:orthord}}\LIPICS{is also deferred to the full paper}.


\PAPER{\paragraph{Techniques.}}
\LIPICS{\subparagraph{Techniques.}}
Various tree-drawing techniques have been identified 
in the large body of
previous work, and
we will certainly draw upon some of these existing
techniques in our new algorithms---in particular, the 
use of ``skewed'' centroids for divide-and-conquer in trees
(see Section~\ref{sec:general} for the definition), and
height--width tradeoffs to obtain
better area bounds.

However, as
the unusual bounds would suggest,
our $n2^{O(\sqrt{\log\log n\log\log\log n})}$ 
and our $n2^{O(\log^*n)}$ algorithms will require
new forms of recursion and bootstrapping.

Our $n2^{O(\sqrt{\log\log n\log\log\log n})}$ result
for arbitrary trees requires novelty not just in fancier
recurrences, but also in geometric insights.
All existing divide-and-conquer algorithms for tree drawings
divide a given tree into subtrees and recursively draw
different subtrees in different,
disjoint axis-aligned bounding boxes.  We will depart
from tradition and draw
some parts of the tree in distorted grids inside narrow sectors,
which are remapped to regular grids through affine transformations
every time we bootstrap.  The key is a geometric observation that
any two-dimensional convex set (however narrow) containing a large number of integer points
must contain a large subset of integer points forming a grid after affine transformation (with unspecified aspect ratio).
The proof of the observation follows from well known facts
about lattices and basis reduction (by Gauss)---a touch of elementary number theory suffices.  We are not aware of previous applications
of this geometric observation, which seems potentially useful for
graph drawing on grids in general.

Our $n2^{O(\log^*n)}$ result is noteworthy,
because occurrences of iterated logarithm are rare in graph drawing 
(to be fair, we should mention that it has appeared before in
one work by Shin et al.~\cite{ShiIPL}, on
poly-line orthogonal drawings of binary trees with $O(1)$ bends per edge).  We realize that more can be gained from
the recursion in the previous $O(n\log\log n)$ algorithm,
by bootstrapping.  This requires a careful setup of 
the recursive subproblems,
and constant switching of $x$ and $y$ (width and height) every time
we bootstrap.  (The author is reminded of an algorithm by
Matou\v sek~\cite{Mat} on a completely different problem, Hopcroft's
problem, where iterated logarithm arose due to constant
switching of points and lines by duality at each level of recursion.)

Our $n2^{O(\sqrt{\log n})}$ result for orthogonal order-preserving
drawings  has the largest quantitative improvement compared to previous results, but actually requires the least originality in techniques.
We use the exact same form of recursion as in an earlier algorithm of
Chan~\cite{SODA99} for non-orthogonal upward order-preserving drawings, although the new algorithm requires trickier details.

\section{Straight-Line Upward Drawings of Arbitrary Trees}
\label{sec:general}

In this section, we consider arbitrary (rooted) trees
and describe our first algorithm to produce straight-line upward
drawings with $o(n\log n)$ area.  It serves as 
a warm-up to the further improved algorithm in Section~\ref{sec:general2}
when upwardness is dropped.

\subsection{Preliminaries}

We begin with some basic number-theoretic and tree-drawing
facts.  The first, on the denseness of relatively prime pairs,
is well known:

\begin{fact}\label{fact:coprime}
There are $\Omega(AB)$ relatively prime pairs 
in $\{1,\ldots,A\}\times\{\down{B/2}+1,\ldots,B\}$.
\end{fact}
\PAPER{
\begin{proof}
The number of pairs in $\{1,\ldots,A\}\times\{1,\ldots,B\}$ that
are \emph{not\/} relatively prime is 
\[
 \le\ \sum_{\mbox{\scriptsize prime }p}\down{\frac Ap}\down{\frac Bp}
 \ \le\
 AB\sum_{\mbox{\scriptsize prime }p} \frac1{p^2}\ <\ 0.453 AB,
\]
whereas the total number of pairs in $\{1,\ldots,A\}\times\{\down{B/2}+1,\ldots,B\}$ is
 $\ge 0.5AB$.
\end{proof}
}

Next, we consider drawing trees not on the integer grid but on a user-specified set of points.
We note that any point set of near linear size that is not 
too degenerate is ``universal'', in the sense that it can be used to
draw any tree.

\begin{fact}\label{fact:universal}
Let $P$ be a set of $(\ell-1) n-\ell+2$ points in the plane, with no $\ell$ points
lying on a common line.
Let $T$ be a tree of size $n$.  
Then $T$ has a straight-line upward drawing
where all vertices are drawn in $P$.
\end{fact}
\begin{proof}
%
We describe a straightforward recursive algorithm:
Let $n_1,\ldots,n_d$ be the sizes of the subtrees $T_1,\ldots,T_d$
at the children of the root $v_0$,
with $\sum_{i=1}^d n_i = n-1$.
Place $v_0$ at the highest point $p_0$ of $P$ (in case of ties, prefer the leftmost highest point).
Form $d$ disjoint sectors with apex at $p_0$,
so that the $i$-th sector $S_i$ contains between
$(\ell-1) n_i-\ell+2$ and $(\ell-1)n_i$ points of $P-\{p_0\}$.
This is possible since any line through $p_0$ contains at most $\ell-2$  points of $P-\{p_0\}$, and $\sum_{i=1}^d (\ell-1) n_i = (\ell-1)(n-1) = |P-\{p_0\}|$.
For each $i=1,\ldots,d$, recursively draw $T_i$ using $(\ell-1) n_i-\ell+2$ points
of $P\cap S_i$.  Lastly, draw the edges from $v_0$
to the roots of the $T_i$'s (these edges create no crossings since the roots are drawn at the highest points of $P$ in their respective sectors).  
The base case $n=1$ is trivial.
%
\end{proof}

The following is a slight generalization of the
standard algorithm (mentioned in the introduction)
for straight-line upward drawings of
general trees with width $O(n)$ and height $O(\log n)$.
We note that the algorithm can draw any tree
on any point set that ``behaves'' like an $n\times\up{\log n}$
grid.

\fig{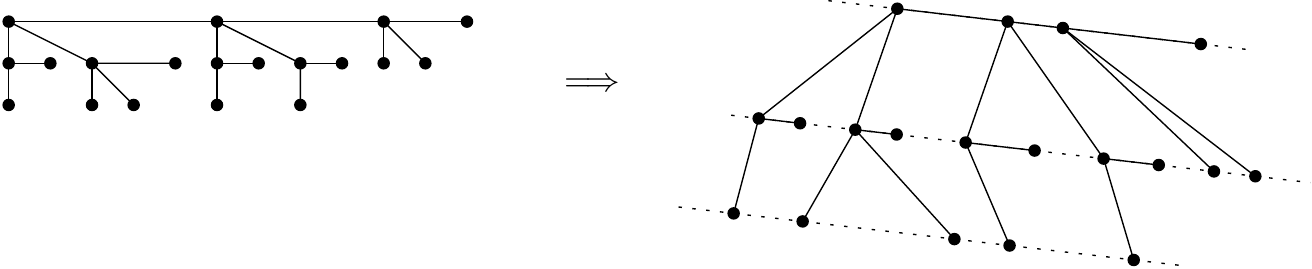}{\PAPER{0.8}\LIPICS{0.7}}{The drawing in Fact~\ref{fact:pseudogrid}.}

\begin{fact}\label{fact:pseudogrid}
Let $G$ be a set of $\up{\log n}$ parallel (non-vertical)
line segments in the plane. 
Let $P$ be a set of $n\up{\log n}$ points, with
$n$ points lying on each of the $\up{\log n}$ line segments in $G$.
Let $T$ be a tree of size $n$.  
Then $T$ has a straight-line drawing
where all vertices are drawn in $P$, and the root is drawn on
the segment of $G$ whose line has the highest $y$-intercept.

Furthermore, if the segments of $G$ are horizontally separated
(i.e., the $y$-projections are disjoint), the drawing is upward.
\end{fact}
\begin{proof}
Without loss of generality, assume that the segments have negative slope,
and arrange the segments of $G$ in decreasing order of $y$-intercepts.
Apply the standard algorithm to get
a straight-line upward grid drawing of $T$
with width at most $n$ and height at most $\up{\log n}$.
Map the vertices on the $i$-th topmost row of the grid drawing 
to the points on the $i$-th segment of $G$, while preserving the left-to-right ordering of the vertices.
(See Figure~\ref{fig_pseudogrid}.)
The resulting drawing is planar (since each edge is drawn either
on a segment or in the region between two consecutive segments,
and there are no crossings in the region between two consecutive segments).
Note that the drawing is upward if the segments of $G$ are horizontally
separated.
\end{proof}

\subsection{The augmented-star algorithm}\label{sec:star}

\newcommand{\LL}{\overline{L}}

The main difficulty of drawing arbitrary trees
is due to the presence of vertices of large degree.
In the extreme case when the tree is a star of size $n$,
we can produce a straight-line drawing of width $O(A)$ and
$O(n/A)$ for any given $1\le A\le n$, by placing the
root at the origin and placing the remaining vertices
at points with \PAPER{relatively }\LIPICS{co-}prime $x$- and $y$-coordinates,
using Fact~\ref{fact:coprime}.

We first study a slightly more general special case which  we call \emph{augmented stars}, where the input tree
is modified from a star by attaching to each leaf
a small subtree of size \PAPER{at most $s$}\LIPICS{$\le s$}.

\begin{lemma}\label{lem:star}
Let $T$ be a tree of size $n$ such that the subtree at 
each child of the root has size at most $s$.
For any given $n\ge A\ge 1$,
$T$ has a straight-line upward drawing 
with width $O(A\log s)$ and height $O((n/A)\cdot s\log^2 s)$,
where the root is placed at the top left corner
of the bounding box,
and the left side of the box contains no other vertices.
\end{lemma}
\begin{proof}
\newcommand{\PPP}{{\cal P}}
Let $\ell = s\up{\log s}$.
Let $B=\up{c\ell n/A}$ for some constant $c$.
Let $P=\{(x,y)\in\{1,\ldots,A\}\times\{-B,\ldots,-\down{B/2}-1\}: \mbox{$x$ and $y$ are relatively prime}\}$.
By Fact~\ref{fact:coprime}, $|P|=\Omega(AB)$, and so $|P|\ge \ell n$
by making $c$ sufficiently large. 


Let $n_1,\ldots,n_d$ be the sizes of the subtrees $T_1,\ldots,T_d$
at the children of the root $v_0$,
with $\sum_{i=1}^d n_i = n-1$ and $n_i\le s$ for each $i$.
Place $v_0$ at the origin.  
Form $d$ disjoint sectors, where the $i$-th sector $S_i$
contains exactly $\ell n_i$ points of $P$.  This is possible, since
any line through the origin contains at most one point of $P$
and $\sum_{i=1}^d \ell n_i < \ell n\le |P|$.
\PAPER{(See Figure~\ref{fig_star}.) }%
We will draw $T$ using not just the points of $P$, but also scaled copies of these points, up to scaling factor $t:=\up{\log s}$.


\PAPER{
\fig{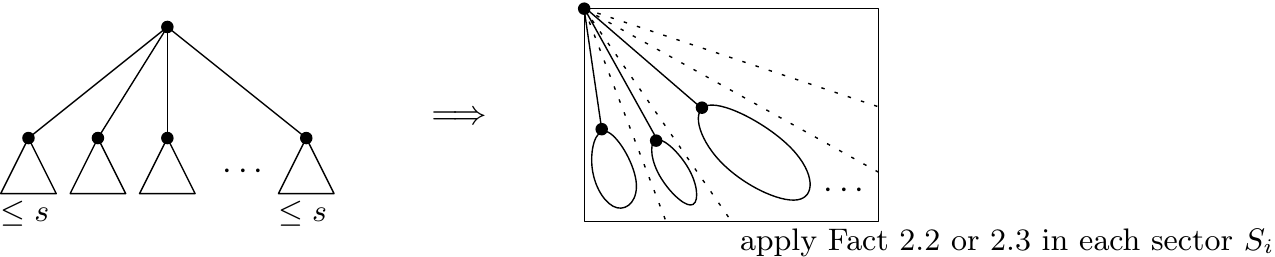}{1.2}{The augmented-star algorithm in Lemma~\ref{lem:star}.}
}

For each $i$, consider two cases, depending on how degenerate $S_i\cap P$
is:
\begin{itemize}
\item {\sc Case 1:} $S_i$ does not contain $\ell$
points of $P$ on a common line.  Here, we can draw $T_i$
using the $\ell n_i > (\ell-1)n_i - \ell+2$ points of $S_i\cap P$ by Fact~\ref{fact:universal}.
\item {\sc Case 2:} $S_i$ contains $\ell$ points of $P$ on a common line $L$.  (Note that $L$ does not pass through the origin,
by definition of $P$.)  
Let $\sigma$ be a horizontal slab of height $B/(2t)$
that contains at least $\ell/t= s$ points of $L\cap S_i\cap P$.
Let $\LL=L\cap S_i\cap \sigma$.
Let $G$ be the set of $t$ line segments $\LL,2\LL,\ldots,t\LL$,
where $\alpha\LL$ denotes the scaled copy of $\LL$ by factor $\alpha$
(with respect to the origin).  Each of the $t=\up{\log s}$
segments of $G$ contain $s$ integer points inside $S_i$,
and the segments are horizontally separated.
Thus, we can
draw $T_i$ using the integer points on $G$ by
Fact~\ref{fact:pseudogrid}.
\end{itemize}
Lastly, draw the edges from $v_0$ to the roots of the $T_i$'s.
The total width is $O(tA)=O(A\log s)$ and the height is $O(tB)=O((n/A)\cdot s\log^2s)$. 
\end{proof}

\subsection{The general algorithm}\label{sec:overall}

\newcommand{\AAA}{\tilde{A}}

We are now ready to present the algorithm for the general case, using the
augmented-star algorithm as a subroutine:

\begin{theorem}\label{thm:general}
For any given $n\ge A\ge 1$,
every tree $T$ of size $n$ has a straight-line upward
drawing with width $O(A+\log n)$ and
height $O((n/\sqrt{A})\log^2 A)$, where the root is placed at the top left corner of
the bounding box.
\end{theorem}
\begin{proof}
We describe a recursive algorithm to draw $T$:
Let $s$ be a fixed parameter with $A\ge\log s$.
Let $v_0$ be the root of $T$, and
define $v_{i+1}$ to be the child of $v_i$ whose subtree is the largest (the resulting root-to-leaf path $v_0v_1v_2\cdots$ 
is called the \emph{heavy path} of $T$).
Let $k$ be the largest index such that the subtree at $v_k$
has size more than $n-A$
(we will call the node $v_k$ the \emph{$A$-skewed centroid}).
Then the total size of the subtrees at the siblings of $v_1,\ldots,v_k$
is at most $A$,
the subtree at $v_{k+1}$ has size at most $n-A$,
and the subtree at each sibling of $v_{k+1}$ has size at most $\min\{n-A,\,n/2\}$.

The drawing of $T$, depicted in Figure~\ref{fig_general}, is
constructed as follows (which includes multiple applications
of the standard algorithm in steps 1 and 3, one application
of the augmented-star algorithm in step~2, and recursive calls
in step~4):

\fig{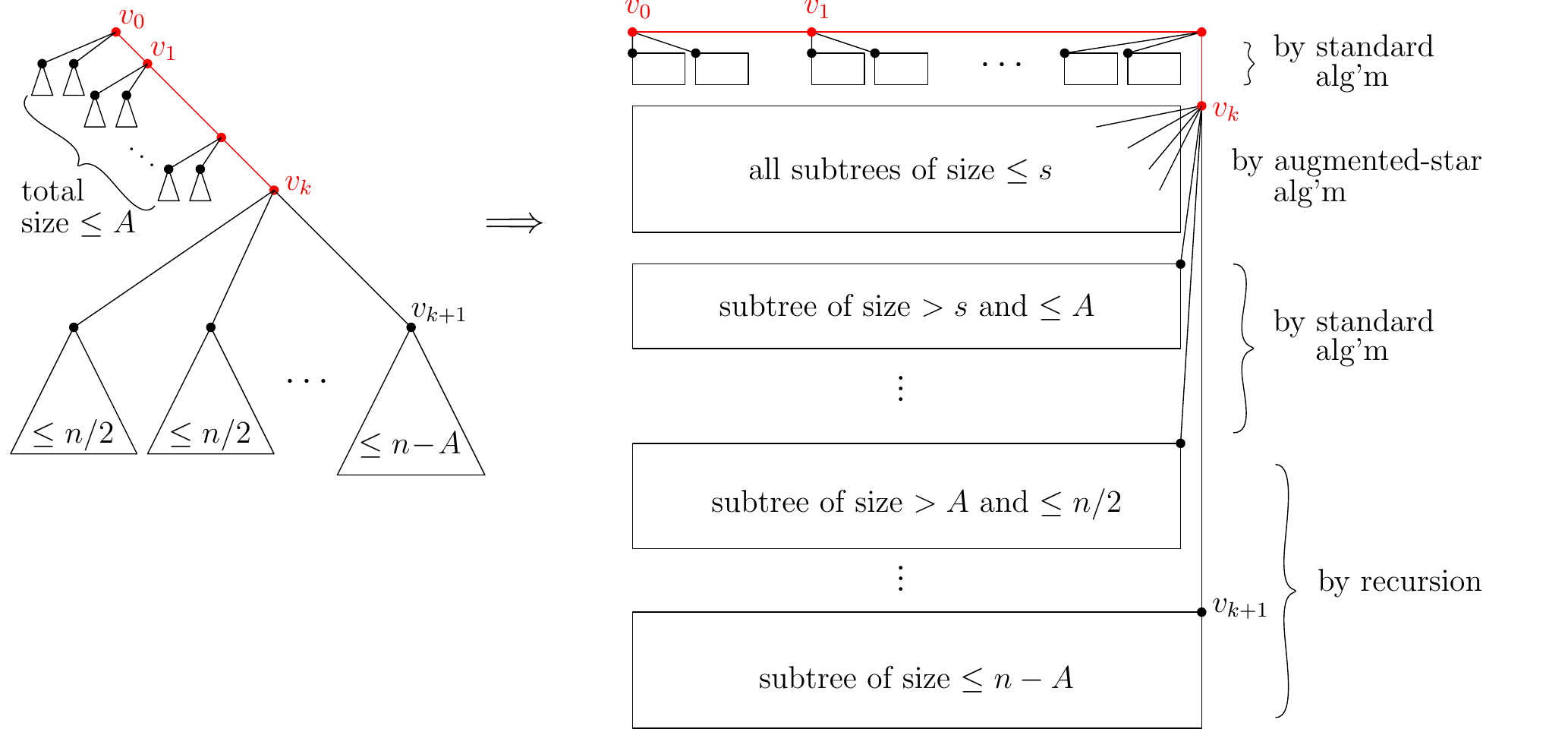}{\PAPER{0.8}\LIPICS{0.7}}{The general algorithm in Theorem~\ref{thm:general}.}

\begin{enumerate}
\item
Draw the subtrees at the siblings of $v_1,\ldots,v_k$
by the standard algorithm.
Stack these drawings horizontally. 
Since these subtrees have total size at most $A$,
the drawing so far has total width $O(A)$ and height $O(\log A)$.
\item
Draw the subtrees at the children of $v_k$ that have
\emph{size $\le s$},
together with the edges from $v_k$ to the roots of these subtrees, 
by the augmented-star algorithm in Lemma~\ref{lem:star} with parameter $\AAA=\up{A/\log s}$.  By reflection, make $v_k$ lie on the top-right corner of its corresponding bounding box.  Place the drawing
below the drawings from step~1.
This part has width $O(\AAA\log s) = O(A)$ and height $O((n'/\AAA)\cdot s\log^2s) = O((n'/A)\cdot s\log^3 s)$ where
 $n'$ is the total size of these subtrees.

(Note that if $n'\le A$, we can
just use the standard algorithm with width $O(A)$ and height $O(\log A)$ for this step.)
\item
Draw the subtrees at the children of $v_k$ that have
\emph{size $> s$ and $\le A$}, by the standard
algorithm.
By reflection, make the roots lie on the top-right corners of their respective bounding boxes.
Stack these drawings vertically, underneath the drawing from step~2.
This part has width $O(A)$ and height $O(\mbox{(number of these subtrees)} \cdot \log A)\le O((n''/s) \cdot \log A)$,
where $n''$ is the total size of these subtrees.
\item
Recursively draw the subtrees at the children of $v_k$ that have 
\emph{size $>A$}.
By reflection, make the roots lie on the top-right corners of their respective bounding boxes.
Stack these drawings vertically, underneath the drawings from step~3.
Put the drawing of the subtree at $v_{k+1}$ 
at the bottom. 
\end{enumerate}

The special case $k=1$ is similar, except that we place
$v_k$ on the left, and so do not reflect in steps 2--4.
The special case $k=0$ is also similar, but bypassing step~1.

The overall width satisfies the following recurrence
\PAPER{
\[ W(n)\ \le\ \max\{ O(A),\ W(n/2)+1,\ W(n-A) \},
\]
}\LIPICS{%
$W(n)\ \le\ \max\{ O(A),\ W(n/2)+1,\ W(n-A) \},$
}
which solves to $W(n)=O(A + \log n)$.

The overall height satisfies the following recurrence
\PAPER{
\[ H(n)\ \le\ \sum_{i=1}^m H(n_i) \:+\: c(\log A + (n'/A)s\log^3 s + (n''/s)\log A)
\]
}\LIPICS{
\[ H(n)\ \le\ \mbox{$\sum_{i=1}^m H(n_i) \:+\: c(\log A + (n'/A)s\log^3 s + (n''/s)\log A)$}\]
}
for some $n',n'',m,n_1,\ldots,n_m$ with $n'+n''+\sum_i n_i\le n$, $n_i\le n-A$,
and $n_i\ge A$, for some constant~$c$.\PAPER{\footnote{
Constants $c$ in different proofs may be different.
}}
It is straightforward to verify by induction%
\PAPER{%
\footnote{
Alternatively, one can see the solution directly without induction:
The contribution of the $(n'/A)s\log^3 s + (n''/s)\log A$ terms cleary sums to at most $(n/A)s\log^3 s + (n/s)\log A$.
The contribution of the first $\log A$ term sums to at most
$(2n/A-1)\log A$, because the number of nodes in the recursion
tree is at most $2n/A-1$.  This is because we can charge at least $A$
units to each leaf and each degree-1 node of the recursion tree
in such a way that the total number of charges
is at most $n$, implying that the number of leaves and degree-1
nodes is  at most $n/A$. The number of nodes of degree at least 2 is
at most the number of leaves minus 1.
}
}
%
that%
\PAPER{
\[ H(n)\ \le\ c((2n/A-1)\log A + (n/A)s\log^3 s + (n/s)\log A). \]
}\LIPICS{
$H(n)\ \le\ c((2n/A-1)\log A + (n/A)s\log^3 s + (n/s)\log A).$
}
(The constraint $n_i\le n-A$ is needed in the $m=1$ case.) %
%
%
%
%
Choosing $s=\Theta(\sqrt{A}/\log A)$ to balance the last two terms gives the height bound in the theorem. 
\end{proof}


Finally, choosing $A=\up{\log n}$ gives\PAPER{: 
\begin{corollary}
Every tree of size $n$ has a straight-line
upward drawing with area $O(n\sqrt{\log n}\log^2\log n)$.
\end{corollary}
}\LIPICS{
area $O(n\sqrt{\log n}\log^2\log n)$.
}

\PAPER{
\noindent{\em Remark.} 
No attempt has been made to improve the minor $\log\log n$ factors.
}

\section{Straight-Line Drawings of Arbitrary Trees}
\label{sec:general2}

To obtain still better area bounds for straight-line non-upward
drawings of arbitrary trees, the idea is to bootstrap: we show how to use
a given general algorithm to obtain an improved augmented-star algorithm,
which in turn is used to obtain an improved general algorithm.
In order to bootstrap, we need to identify large
grid substructures inside each sector in the augmented-star algorithm.
This requires an interesting geometric observation about lattices, described in
the following subsection.

\subsection{An observation about lattices}

\newcommand{\ff}{f_0}
\newcommand{\gggg}{g_0}
\newcommand{\uu}{\textbf{u}}
\newcommand{\vv}{\textbf{v}}
\newcommand{\Z}{\mathbb{Z}}
\newcommand{\diam}{\textrm{diam}}

A \emph{two-dimensional lattice} is a set of the form
$\Lambda=\{i\uu+j\vv : i,j\in\Z\}$ for some vectors $\uu,\vv\in\R^2$.
The vector pair $\{\uu,\vv\}$ is called a \emph{basis} of $\Lambda$.

In this paper, we use the term
\emph{$a\times b$ affine grid\/} to refer to a set of the form
$\{i\uu+j\vv: i\in\{x_0+1,\ldots,x_0+a\},\,j\in\{y_0+1,\ldots,y_0+b\}\}$
for some vectors $\uu,\vv\in\R^2$ and some $x_0,y_0\in\R$.
In other words, it is a set that is equivalent to the
regular $a\times b$ grid $\{1,\ldots,a\}\times\{1,\ldots,b\}$ after
applying some affine transformation.

The following observation is the key (see Figure~\ref{fig_lattice}).
The author is not aware of any references of this specific statement
(but would not be surprised if this was known before).


\fig{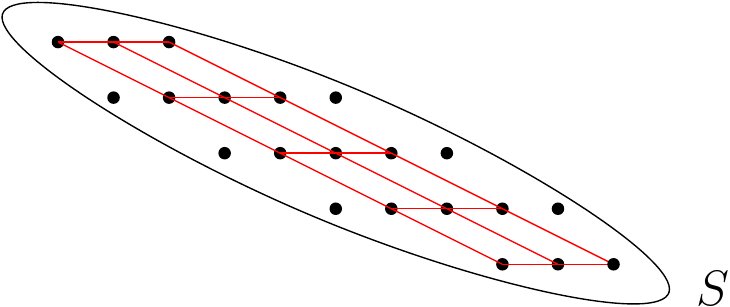}{\PAPER{0.7}\LIPICS{0.6}}{Observation~\ref{obs:lattice}:
A convex set that contains many lattice points
must contain a large affine grid in the lattice.}

\begin{obs}\label{obs:lattice}
If a convex set $S$ in the plane contains $n$ points from a lattice $\Lambda$,
then $S\cap\Lambda$ contains an $a\times b$ affine grid for some $a$ and
$b$ with $ab=\Omega(n)$. 
\end{obs}
\begin{proof}
First, apply an affine transformation to make $S$ \emph{fat}, i.e.,
$D^-\subset S\subset D^+$ for some disks $D^-$ and $D^+$ 
 with $\diam(D^-)= \Omega(\diam(D^+))$.
(This follows immediately from well-known properties
of the \emph{L\"owner--John ellipsoid}; or
see~\cite{AgHaVa,BarHar} for simple, direct algorithms.)

After the transformation, $\Lambda$
is still a lattice.  It is well known that there
exists a basis $\{\uu,\vv\}$ for $\Lambda$ 
satisfying
$60^\circ\le \angle(\uu,\vv)\le 120^\circ$.
(A {\em Gauss-reduced basis\/} satisfies this property;
for example, see \cite[Section 27.2]{VazBOOK}.)

Let $R^+$ be the smallest rhombus containing $D^+$, with sides parallel to $\uu$
and $\vv$.
Let $R^-$ be the largest rhombus $R^-$ contained in $D^-$, with
sides parallel to $\uu$ and $\vv$.
Then $R^+$ and $R^-$ have side lengths $r^+=O(\diam(D^+))$
and $r^-=\Omega(\diam(D^-))$ respectively,
since $\angle(\uu,\vv)$ is bounded away from $0^\circ$ or $180^\circ$.  It follows that $r^-=\Omega(r^+)$.

Now, $S\cap\Lambda\subset R^+\cap\Lambda$ is contained in an $\up{r^+/\|\uu\|}
\times \up{r^+/\|\vv\|}$ affine grid.
Thus, 
\PAPER{$$n\:\le\: \up{r^+/\|\uu\|}\cdot \up{r^+/\|\vv\|}.$$}%
\LIPICS{$n\:\le\: \up{r^+/\|\uu\|}\cdot \up{r^+/\|\vv\|}.$}

On the other hand,
$S\cap\Lambda\supset R^-\cap\Lambda$ contains an $\down{r^-/\|\uu\|}
\times \down{r^-/\|\vv\|}$ affine grid, with%
\PAPER{
$$\down{r^-/\|\uu\|}
\times \down{r^-/\|\vv\|}\:=\:\Omega(\up{r^+/\|\uu\|}\cdot \up{r^+/\|\vv\|})\:=\:\Omega(n)$$
}\LIPICS{
$\down{r^-/\|\uu\|}
\times \down{r^-/\|\vv\|}\:=\:\Omega(\up{r^+/\|\uu\|}\cdot \up{r^+/\|\vv\|})\:=\:\Omega(n)$
}%
points, assuming that $\|\uu\|,\|\vv\|\le r^-$.

This almost completes the proof.
It remains to address the special case when $\|\uu\|>r^-$
(the case $\|\vv\|>r^-$ is similar).  Here,
$S\cap\Lambda\subset R^+\cap\Lambda$ is contained in an
$O(1)\times \up{r^+/\|\vv\|}$  affine grid.  Some row of the grid
must contain $\Omega(n)$ points of $S\cap\Lambda$.
The row is a $1\times \Omega(n)$ affine grid.
\end{proof}

\subsection{Improved augmented-star algorithm}

\newcommand{\GEN}{{\cal G}_0}

We first show how to use a given general algorithm $\GEN$ to
obtain an improved algorithm for the augmented-star case:



\begin{lemma}\label{lem:star2}
Suppose we are given a \emph{general algorithm} $\GEN$ that takes as input
any $n\ge A\ge g_0(n)$ and any tree of size $n$, and outputs a straight-line
drawing of width at most $A$ and height at most
$(n/A)\ff(A)$,
where the root is drawn at the top left corner
of the bounding box.
Here,  $\ff$ and $g_0$ are some
increasing functions  satisfying $\ff(n)\ge g_0(n)$.

Then we can obtain an \emph{improved augmented-star algorithm} that
takes as input any $n\ge A\ge 1$ and
a tree of size $n$ such that the subtree at 
each child of the root has size at most $s$,
and outputs a straight-line drawing 
with width $O(A\log s)$ and height $O((n/A)\cdot \ff(s)\log s)$,
where the root is placed at the top left corner
of the bounding box,
and the left side of the box contains no other vertices.
\end{lemma}
\begin{proof}
\newcommand{\PPP}{{\cal P}}
Let $\ell=c\ff(s)$ for some constant $c$.
Let $B=\up{c\ell n/A}$.
Let $P=\{(x,y)\in\{1,\ldots,A\}\times\{-B,\ldots,-1\}: \mbox{$x$ and $y$ are relatively prime}\}$.
By Fact~\ref{fact:coprime}, $|P|=\Omega(AB)$,
and so $|P|\ge\ell n$ by making $c$ sufficiently large.

Let $n_1,\ldots,n_d$ be the sizes of the subtrees $T_1,\ldots,T_d$
at the children of the root $v_0$,
with $\sum_{i=1}^d n_i = n-1$ and $n_i\le s$ for each $i$.
Place $v_0$ at the origin.  
Form $d$ disjoint sectors, where the $i$-th sector $S_i$
contains exactly $\ell n_i $ points of $P$.  This is possible, since
any line through the origin contains at most one point of $P$
and $\sum_{i=1}^d  \ell n_i < \ell n\le |P|$.

Take a fixed $i$.  Applying Observation~\ref{obs:lattice}
to the convex set $S_i\cap ((0,A]\times [-B,0))$, we see that
$S_i\cap(\{1,\ldots,A\}\times \{-B,\ldots,-1\})$ must contain
an $a\times b$ affine grid for some $a$ and $b$
with $ab=\Omega(\ell n_i)$.  
Note that 
$b\ge (n_i/a)\ff(s)$ by making $c$ sufficiently large.
Consider two cases:
\begin{itemize}
\item {\sc Case 1:} $g_0(n_i)\le a\le n_i$. 
Here, we can draw $T_i$ in the $a\times b$ affine
grid by the given algorithm~$\GEN$, after applying
an affine transformation to convert to a standard integer $a\times b$ grid.
Note that planarity and straightness are preserved under the transformation (but
not upwardness).
The root of $T_i$ can be placed at the highest corner
of the grid. 
\item {\sc Case 2:} $a>n_i$ or $a< g_0(n_i)$.
Note that in the latter subcase, $b\ge (n_i/a)\ff(s) \ge (n_i/a)g_0(n_i)\ge n_i$.
In either subcase,
$S_i$ contains $n_i$ points of $P$
on a common line $L$.  (Note that $L$ does not pass through the origin,
by definition of $P$.)  
Let $t=\up{\log s}$ and $\LL=L\cap S_i$.
Let $G$ be the $t$ line segments $\LL,2\LL,\ldots,t\LL$.  Then each of the $t=\up{\log s}$
segments of $G$ contain $n_i$ integer points inside $S_i$.
Thus, we can
draw $T_i$ using the integer points on $G$ by
Fact~\ref{fact:pseudogrid}.
The root is placed on the highest segment of $G$.
\end{itemize}
Lastly, draw the edges from $v_0$ to the roots of the $T_i$'s.
%
The total width is $O(tA)=O(A\log s)$ and height is $O(tB)=O((n/A)\cdot \ff(s)\log s)$.
\end{proof}

\subsection{Improved general algorithm}

Using the improved augmented-star algorithm, we can then obtain
an improved general algorithm,
by following the same approach as in the proof of Theorem~\ref{thm:general}, except with Lemma~\ref{lem:star}
replaced by the improved Lemma~\ref{lem:star2} in step~2.
The same analysis shows the following:

\begin{theorem}\label{thm:general2}
Suppose we are given a \emph{general algorithm} $\GEN$ that takes as input
any $n\ge A\ge g_0(n)$ and any tree of size $n$, and outputs a straight-line
drawing of width at most $A$ and height at most
$(n/A)\ff(A)$,
where the root is drawn at the top left corner
of the bounding box.
Here,  $\ff$ and $g_0$ are some
increasing functions  satisfying $\ff(n)\ge g_0(n)$.

Then we can obtain an \emph{improved general algorithm} that
takes as input any $n\ge A\ge\log s$ and any tree of size $n$,
and outputs a straight-line upward
drawing with width $O(A+\log n)$ and
height $O((n/A)\log A + (n/A)\ff(s)\log^2s + (n/s)\log A)$, where
the root is placed at the top left corner of
the bounding box.
\end{theorem}

Assume inductively that there is a general algorithm $\GEN$ satisfying the assumption of the above theorem with
$\ff(A)=C_j A^{1/j}\log^j A$ and $g_0(n)=c_0\log n$
for some $C_j$ and $c_0$.  For $j=1$, this follows from the
standard algorithm, which has logarithmic width and linear height after swapping $x$ and $y$, with $C_1,c_0=O(1)$.

Choosing $s=\up{A^{j/(j+1)}/\log^j A}$ to balance the last two terms
in the above theorem
gives a width bound of $O(A+\log n)$
and height bound of 
$$O((n/A)\log A + (n/A)C_j s^{1/j}\log^{j+2}s + (n/s)\log A)\ =\ O(C_j (n/A) A^{1/(j+1)}\log^{j+1}A).$$
By setting $\AAA = c_0A$ and $C_{j+1}=O(1)\cdot C_j$, with
a sufficiently large absolute constant $c_0$,
the width is at most $\AAA$ and
the height is at most $C_{j+1}(n/\AAA)\AAA^{1/(j+1)}\log^{j+1}\AAA$
for any $n\ge \AAA\ge c_0\log n$.
We have thus obtained a new general algorithm with
$\ff(\AAA)=C_{j+1}\AAA^{1/(j+1)}\log^{j+1}\AAA$
and $g_0(n)=c_0\log n$.


Note that $C_j=2^{O(j)}$.
For the best bound, we choose
a nonconstant $j=\Theta(\sqrt{\log A/\log\log A})$ so that
$\ff(A)=2^{O(j)}A^{1/j}\log^j A = 2^{O((\log A)/j + j\log\log A)}=
2^{O(\sqrt{\log A\log\log A})}$, yielding:

\begin{corollary}
For any given $n\ge A\ge \log n$, every tree of size $n$ has a straight-line
drawing with width $O(A)$ and
height $(n/A) 2^{O(\sqrt{\log A\log\log A})}$.
\end{corollary}

Finally, choosing $A=\up{\log n}$ gives\PAPER{ :
\begin{corollary}
Every tree of size $n$ has a straight-line
drawing with area 
$n2^{O(\sqrt{\log\log n\log\log\log n})}$.
\end{corollary}
}\LIPICS{
area $n2^{O(\sqrt{\log\log n\log\log\log n})}$.
\par\medskip
}

\noindent {\em Remarks.}
It is straightforward to implement the algorithms in Section~\ref{sec:general} and this section in polynomial time.
\PAPER{\par}%
One open question is whether the improved bound holds for
upward drawings.  
Another open question is whether further improvements are possible
if we allow poly-line drawings.

\section{Straight-Line Orthogonal Drawings of Binary Trees}
\label{sec:logstar}

In this section, we consider binary trees and describe algorithms to
produce straight-line orthogonal (non-upward) drawings.
We improve previous algorithms with $O(n\log\log n)$ area
by Shin, Kim, and Chwa~\cite{ShKiCh} and
Chan et al.~\cite{GD96}.  The idea is (again) to bootstrap.

Given a binary tree $T$ and two distinct vertices $u$ and $v$,
such that $v$ is a descendant of $u$ but not an immediate child of $v$,
the \emph{chain\/} from $u$ to $v$ is defined to be the
subtree at $u$ minus the subtree at $v$.  (To explain the terminology,
note that the chain consists of the path from $u$ to the parent of $v$,
together with a sequence of subtrees attached to the nodes of this path.)  We show how to use a given algorithm for drawing chains to obtain a general algorithm for drawing trees, which together with the given chain algorithm is used to obtain an improved chain algorithm.

\subsection{The general algorithm}

\newcommand{\CHAIN}{{\cal C}}

Given a chain algorithm $\CHAIN_0$,
we can naively use it to
draw the entire tree, since a tree can be viewed as a chain
from the root to an artificially created leaf.
We first show how to use a given chain algorithm $\CHAIN_0$
to obtain a general algorithm that achieves \emph{arbitrary
width--height tradeoffs}.  This is done by adapting
previous algorithms~\cite{ShKiCh,GD96}.

\begin{lemma}\label{lem:tradeoff}
Suppose we are given a \emph{chain algorithm} $\CHAIN_0$ that
takes as input any binary tree and a chain from $v_0$ to 
$v_k$ where the size of the chain is $n$, 
and outputs a straight-line
orthogonal drawing of the chain with width at most $W_0(n)$ and height
at most $H_0(n)$,
where $v_0$ is placed at the top
left corner of the bounding box, and the parent of $v_k$ is placed at the
bottom left corner of the box. 
Here, $W_0(n)$ and $H_0(n)$ are
increasing functions. 

Then we can obtain a \emph{general algorithm} that takes as input 
$n\ge A\ge 1$ and any binary tree $T$ of size $n$,
and outputs
a straight-line orthogonal drawing with
width $O(W_0(A)+\log n)$ and height $O((n/A)H_0(A))$, where 
the root is placed at the top left corner of the bounding box.
\end{lemma}
\begin{proof}
We describe a recursive algorithm to draw $T$:
Let $v_0v_1v_2\cdots$ be the heavy path,
and $v_k$ be the $A$-skewed centroid, as in the
proof of Theorem~\ref{thm:general}.
Then the chain from $v_0$ to $v_k$ has size at most $A$,
 the subtree at $v_{k+1}$ has size at most $n-A$,
and the subtree at the sibling of $v_{k+1}$ has size at most
$\min\{n-A,n/2\}$.

\fig{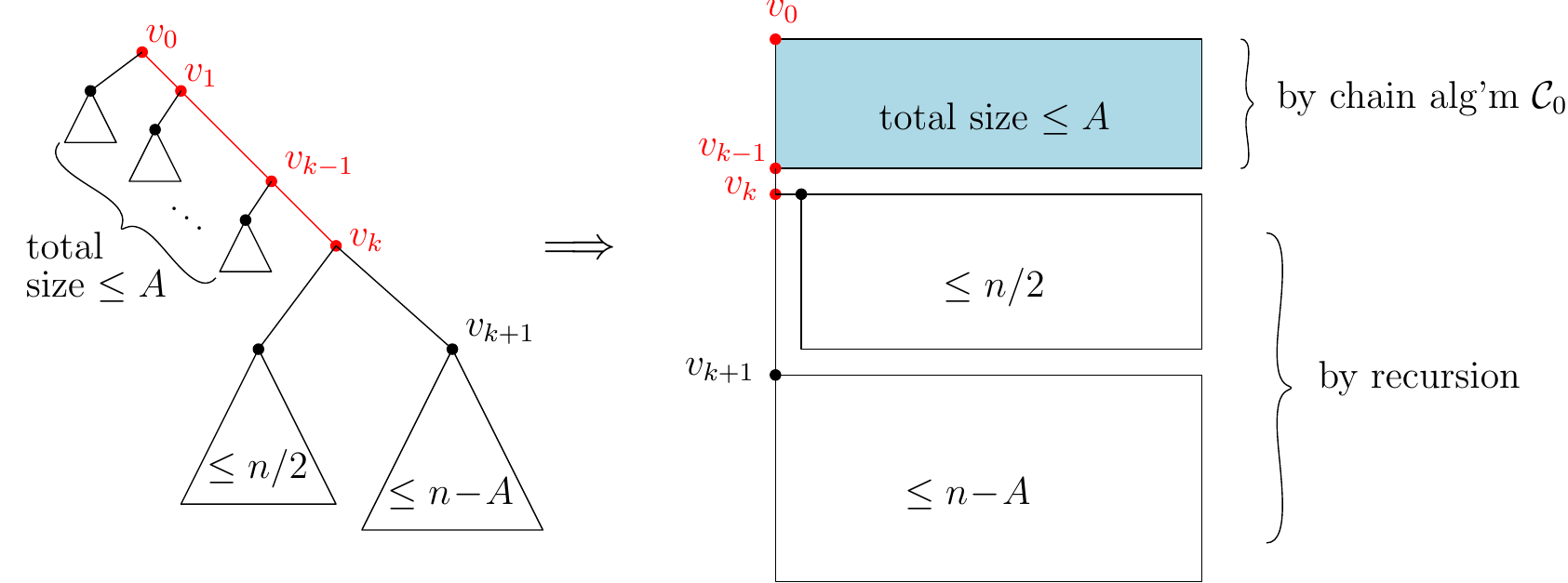}{\PAPER{0.8}\LIPICS{0.7}}{The general algorithm in Lemma~\ref{lem:tradeoff} for orthogonal drawings.}

The drawing of $T$, depicted in Figure~\ref{fig_tradeoff}, is
constructed as follows:

\begin{enumerate}
\item 
Draw the chain from $v_0$ to $v_k$ by the given algorithm $\CHAIN_0$, with width at most $W_0(A)$ and height at most
$H_0(A)$.  
\item Recursively draw the subtrees at the two children
of $v_k$.  
Stack the two drawings vertically, underneath the drawing from step~1.
Put the drawing of the subtree at $v_{k+1}$ at the bottom.
\PAPER{\par}
(Note that if any of these subtrees has size at most $A$, we can just use
algorithm $\CHAIN_0$ with width at most $W_0(A)$ and height
at most $H_0(A)$.)
\end{enumerate}

The special case $k=1$ is similar, except that in step~1
we can just apply algorithm $\CHAIN_0$ to draw the subtree at
the sibling of $v_1$, and connect $v_0$ to $v_k$ directly.  The special case $k=0$ is also similar, but bypassing step~1.

The overall width satisfies the recurrence%
\PAPER{
\[ W(n)\ \le\ \max\{ O(W_0(A)),\ W(n/2)+1,\ W(n-A) \},
\]
}\LIPICS{
$W(n)\ \le\ \max\{ O(W_0(A)),\ W(n/2)+1,\ W(n-A) \},$%
}
which solves to $W(n)=O(W_0(A) + \log n)$.

The overall height satisfies the recurrence%
\PAPER{
\[ H(n)\ \le\ \sum_{i=1}^m H(n_i) + cH_0(A)
\]
}\LIPICS{
$H(n)\ \le\ \sum_{i=1}^m H(n_i) + cH_0(A)$%
}
for some $m,n_1,\ldots,n_m$ with $m\le 2$, $\sum_i n_i\le n$, $n_i\le n-A$,
and $n_i\ge A$, 
for some constant~$c$.
The recurrence solves to $H(n)\le\ c(2n/A-1)H_0(A)$
(similarly to the proof of Theorem~\ref{thm:general}). 
\end{proof}

\subsection{The improved chain algorithm}

\fig{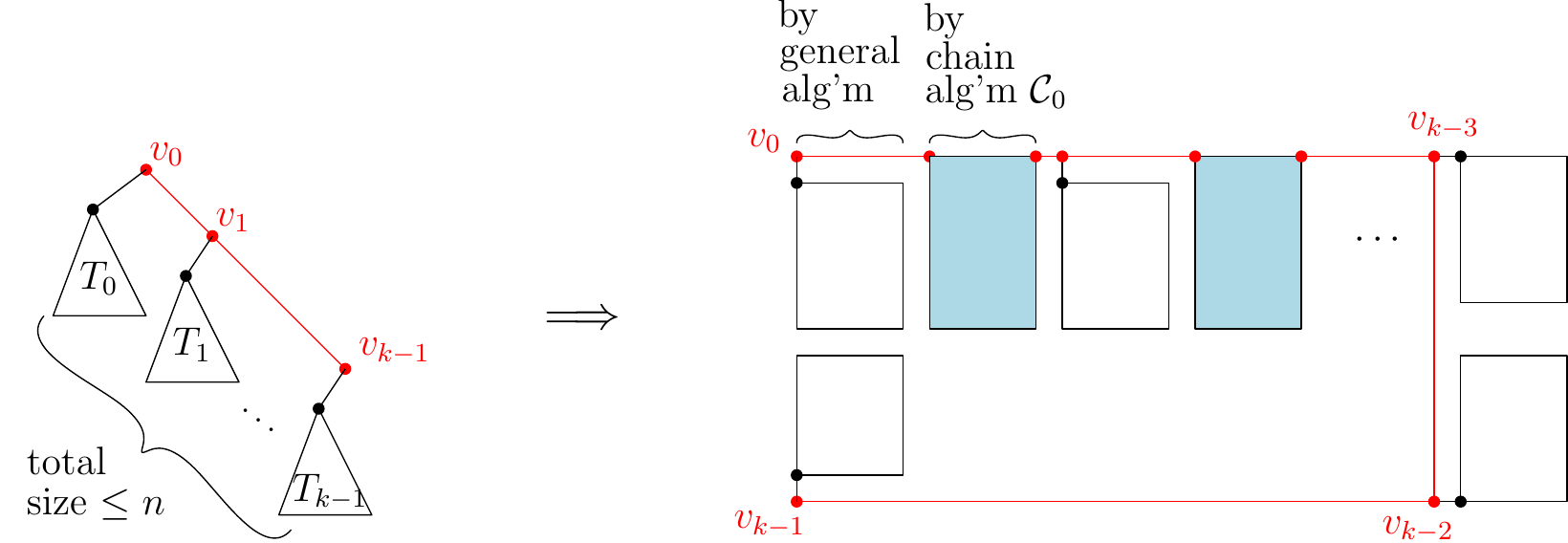}{\PAPER{0.8}\LIPICS{0.7}}{The improved chain algorithm 
in Theorem~\ref{thm:chain} for orthogonal drawings.}

Using both the general algorithm from Lemma~\ref{lem:tradeoff}
and the given chain algorithm $\CHAIN_0$, we describe an improved chain algorithm:

\begin{theorem}\label{thm:chain}
Suppose we are given a \emph{chain algorithm} $\CHAIN_0$ that
takes as input any binary tree and a chain from $v_0$ to 
$v_k$ where the size of the chain is $n$, 
and outputs a straight-line
orthogonal drawing of the chain with width at most $W_0(n)$ and height
at most $H_0(n)$,
where $v_0$ is placed at the top
left corner of the bounding box, and the parent of $v_k$ is placed at the
bottom left corner of the box. 
Here, $W_0(n)$ and $H_0(n)$ are
increasing functions. 

Then we can obtain an \emph{improved chain algorithm} that takes
as input any $n\ge A\ge 1$ and any binary tree and a chain from $v_0$ to $v_k$ where
the size of the chain is $n$,
and outputs a straight-line
orthogonal drawing of the chain with width 
$O((n/A)H_0(A))$ and height
$O(W_0(A)+\log n)$,
where $v_0$ is placed at the top
left corner of the bounding box, and the parent of $v_k$ is placed
at the bottom left corner.
\end{theorem}
\begin{proof}
Let $v_0v_1\cdots v_k$ denote the path from $v_0$ to $v_k$.
Let $T_i$ denote the subtree at the
sibling of $v_{i+1}$.  Let $n_i$ be the size of $T_i$ plus 1.

Divide the sequence $v_0v_1\cdots v_{k-4}$ into subsequences,
where each subsequence is either (i)~a \emph{singleton\/} $v_i$,
or (ii)~a contiguous \emph{block\/} $v_iv_{i+1}\cdots v_\ell$ of length at least 2 with $n_i+n_{i+1}+\cdots+n_\ell\le A$.
By making the blocks maximal, we can
ensure that the number of singletons and blocks is $O(n/A)$.
We add $v_{k-3},\ldots,v_{k-1}$ as 3 extra singletons.
\begin{itemize}
\item For each singleton $v_i$,
draw $T_i$ by the general algorithm in Lemma~\ref{lem:tradeoff}
if $n_i\ge A$, or 
directly by the given algorithm ${\cal C}_0$ if $n_i<A$. 
By swapping $x$ and $y$, the width is $O((n_i/A+1)H_0(A))$
and the height is $O(W_0(A)+\log n)$.  
\item For each block $v_iv_{i+1}\cdots v_\ell$,
draw the subchain from $v_i$ to $v_{\ell+1}$,
which has size at most $A$, by the given algorithm ${\cal C}_0$.
By swapping $x$ and $y$,
the width is $O(H_0(A))$
and the height is $O(W_0(A))$.  
\end{itemize}
All these drawings are stacked horizontally as shown in Figure~\ref{fig_chain},
except for $T_{k-2}$ and $T_{k-1}$, which are placed below and flipped upside-down.

The special cases with $k\le 3$ are simpler: just stack the $O(1)$ drawings
 vertically, with the bottom drawing of $T_{k-1}$ flipped upside-down.

The total width due to singletons is $O(\sum_i (n_i/A+1)H_0(A))=O((n/A)H_0(A))$, and
the total width due to blocks is also $O((n/A)H_0(A))$, because
the number of singletons and blocks is $O(n/A)$.  The overall height is $O(W_0(A)+\log n)$.
\end{proof}

Assume inductively that there is a chain algorithm ${\cal C}_0$ satisfying the assumption of Theorem~\ref{thm:chain} with 
$W_0(n)=C_j (n/\log n)\log^{(j)}n$ 
and $H_0(n) = C_j \log n$ for some $C_j$, where $\log^{(j)}$ denotes the $j$-th iterated logarithm.  For $j=1$, this follows 
by simply applying the standard algorithm to draw the subtrees $T_i$
in the proof of Theorem~\ref{thm:chain},  with $C_1=O(1)$.

Choosing $A=\up{\log n \log\log n/\log^{(j+1)}n}$
in Theorem~\ref{thm:chain}
gives a width bound of
\PAPER{
$$O((n/A)H_0(A))\,=\,O((n/A) C_j\log A)\,=\,
O(C_j (n/\log n)\log^{(j+1)}n)$$
}\LIPICS{%
$O((n/A)H_0(A))\,=\,O((n/A) C_j\log A)\,=\,
O(C_j (n/\log n)\log^{(j+1)}n)$%
}
and a height bound of 
\PAPER{
$$O(W_0(A)+\log n)\,=\,
O(C_j (A/\log A)\log^{(j)}A + \log n)
\,=\,O(C_j \log n).$$
}\LIPICS{%
$O(W_0(A)+\log n)\,=\,
O(C_j (A/\log A)\log^{(j)}A + \log n)
\,=\,O(C_j \log n).$
}
By setting $C_{j+1}=O(1)\cdot C_j$, we have thus obtained
a new chain algorithm with
$W_0(n)=C_{j+1}(n/\log n)\log^{(j+1)}n$ and 
$H_0(n)=C_{j+1}\log n$.

Note that $C_j=2^{O(j)}$. For the best bound, 
we choose a nonconstant $j=\log^* n$, yielding:

\begin{corollary}
Every binary tree of size $n$ has a straight-line orthogonal
drawing with area $n2^{O(\log^*n)}$.
\end{corollary}

\PAPER{
Tradeoffs can then be obtained by one final application of the general algorithm in
Lemma~\ref{lem:tradeoff},
with width $O(W_0(A)+\log n)=
O(C_j(A/\log A)\log^{(j)}A + \log n)$
and height $O((n/A)H_0(A))=O(C_j(n/A)\log A)$.
Setting $\AAA=C_j(A/\log A)$ and  $j=\log^* n$
yields:
 
\begin{corollary}
For any given $\log n\le \AAA \le n/\log n$,
every binary tree of size $n$ has a straight-line
orthogonal drawing with 
width $O(\AAA)$
and height $(n/\AAA) 2^{O(\log^*\AAA)} $.
\end{corollary}

}

\LIPICS{

\small
\bibliographystyle{abbrvurl}
\bibliography{more_trees}

}

\PAPER{ 

\section{Straight-Line Order-Preserving Drawings of Binary Trees}
\label{sec:logstar2}

We now note how to adapt the algorithm
from Section~\ref{sec:logstar} to
straight-line non-orthogonal order-preserving  drawings.
This improves the previous algorithm with $O(n\log\log n)$ area
by Garg and Rusu~\cite{GarRus03}.

The new algorithm follows the same
recursion and analysis as in Section~\ref{sec:logstar}, except that the geometric placement of subtrees is different.
We describe these differences.
In the given chain algorithm $\CHAIN_0$, the output requirement is changed to the following: $v_0$ may be
placed anywhere on the left side of the bounding box, with no other vertices placed on the left side,  and the parent of
$v_k$ may be placed anywhere \emph{on the right side\/} of
the bounding box, with no other vertices on the right side. 
We further require that order is preserved around the parent of $v_k$ 
even if we were to add $v_k$ to the drawing, placed anywhere to the right
of the bounding box.

In the general algorithm,
the requirement is that $v_0$ is
be placed on the left side of the bounding box, with no
vertices placed directly above $v_0$.

The general algorithm in Lemma~\ref{lem:tradeoff} can be modified
as shown in Figure~\ref{fig_tradeoff_new}, with drawings of the two children
of $v_k$ reflected.  This is similar to the previous algorithm
of Garg and Rusu~\cite{GarRus03}.  
The base case $k=1$ is similar,
except that we just use the algorithm $\CHAIN_0$ to draw
the subtree at the sibling of $v_1$ (which may be placed above or
below $v_0$), and connect from $v_0$ to
$v_k$ directly.
The base case $k=0$ is easier, without needing to reflect.

\fig{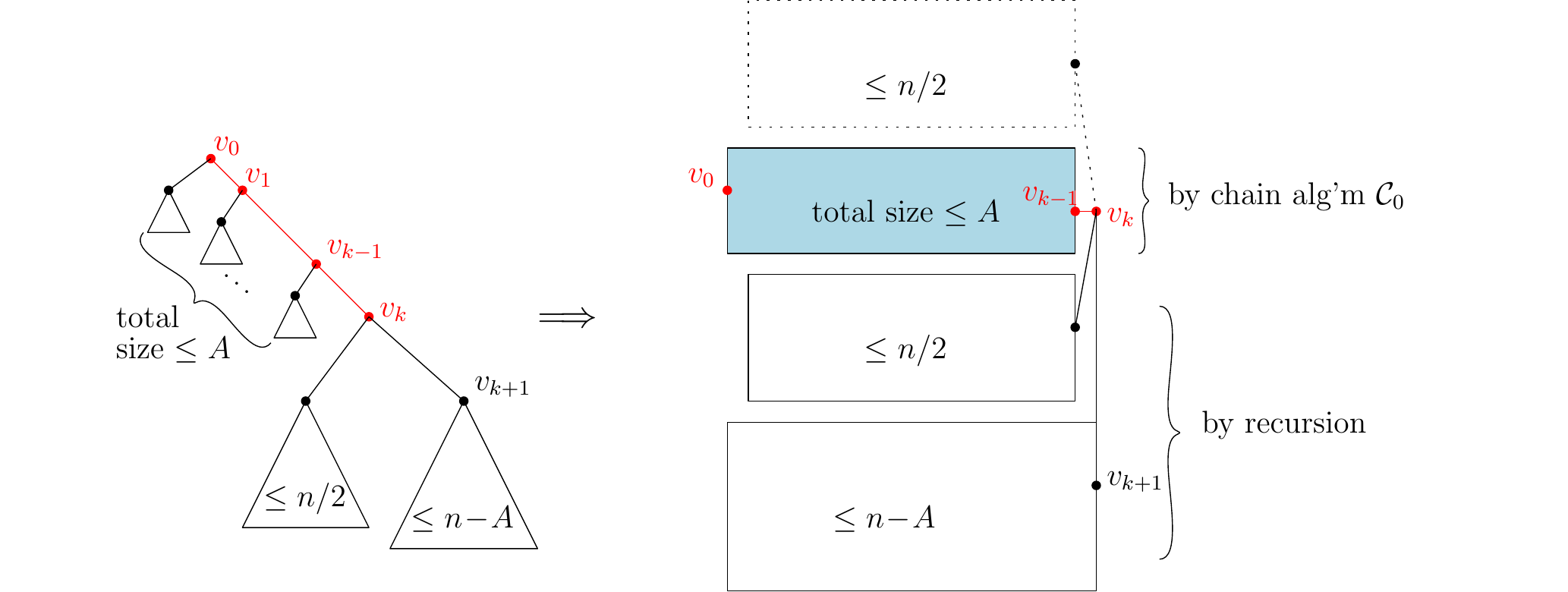}{\PAPER{0.8}\LIPICS{0.7}}{The general algorithm
for order-preserving drawings.  The dotted lines show
an alternative placement of the subtree at $v_{k+1}$'s sibling when $v_{k+1}$ is a left child instead.}

The improved chain algorithm in Theorem~\ref{thm:chain} can be modified
as shown in Figure~\ref{fig_chain_new}, with some drawings flipped
upside-down (besides swapping of $x$ and $y$).  The path $v_0v_1\cdots v_{k-1}$
may now oscillate more in $y$, but the height bound is still
the same within constant factors.

\fig{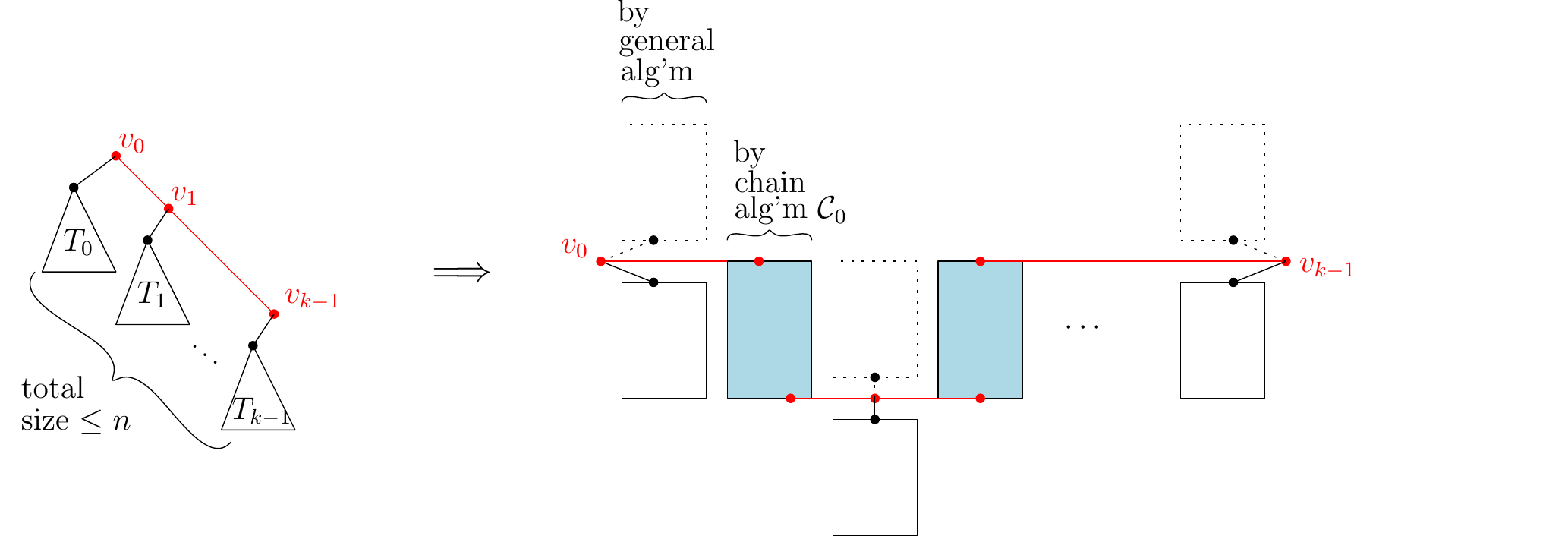}{\PAPER{0.8}\LIPICS{0.7}}{The improved chain algorithm 
for order-preserving drawings.  The dotted lines show
alternative placements of drawings when the order of various siblings is reversed.
}

\begin{corollary}
Every binary tree of size $n$ has a straight-line order-preserving
drawing with area $n2^{O(\log^*n)}$.
\end{corollary}

\PAPER{

\begin{corollary}
Given any $\log n\le \AAA \le n/\log n$,
every binary tree of size $n$ has a straight-line
order-preserving drawing with 
width $O(\AAA)$
and height $(n/\AAA) 2^{O(\log^*\AAA)} $.
\end{corollary}

}

\noindent {\em Remarks.} It is straightforward to implement 
the algorithms in Section~\ref{sec:logstar} and this section to run in linear time.

\LIPICS{
For both the algorithm in Section~\ref{sec:logstar} and this section,
we can obtain height--width tradeoffs by one more application of
the general algorithm in
Lemma~\ref{lem:tradeoff}.
}

These improved results raise the next logical question: is linear area possible for straight-line orthogonal drawings,
or straight-line order-preserving drawings?  Also, could the ideas
here improve the $O(n\log\log n)$ area bound for straight-line upward drawings of binary trees by Shin, Kim, and 
Chwa~\cite{ShKiCh}?

\section{Straight-Line Orthogonal Order-Preserving Drawings of Binary Trees}\label{sec:orthord}

In this section, we consider straight-line (non-upward)
drawings of binary trees that are
both orthogonal and order-preserving.
We improve a previous algorithm by Frati~\cite{Fra07} with
$O(n^{3/2})$ area.
Here, our idea is to adapt an approach by Chan~\cite[Section~5]{SODA99}
for obtaining $n2^{O(\sqrt{\log n})}$ area bounds,
originally designed
for a different class of drawings (straight-line, non-orthogonal, strictly upward, order-preserving).  Our new algorithm follows the
same recursion and analysis
as in \cite{SODA99}, but the geometric placement of subtrees is more
involved.

We describe a recursive algorithm to draw $T$, where the root $v_0$
is placed inside the bounding box, with the requirement that
planarity and order is preserved even if we were to add a new edge
to the drawing, entering $v_0$ horizontally from the right:

Let $A$ be a parameter to be chosen later.
Let $v_0v_1v_2\cdots$ be the heavy path,
and $v_k$ be the $A$-skewed centroid, as defined in the
proof of Theorem~\ref{thm:general}.
Recursively draw the subtrees at the siblings of $v_1,\ldots,v_k$,
as well as the subtrees at the two children of $v_k$.

We assume that $v_1$ is a right child (the other case is symmetric, as explained later).

Let $j$ be the largest index such that $v_j$ is a left child with $j\le k$.  Then $v_jv_{j+1}\cdots v_k$ is
a rightward path.
We put the drawings together in one of the two ways depicted in 
Figure~\ref{fig_orth_ord}, depending on whether
$v_{j-1}$ is a right child or a left child.

\fig{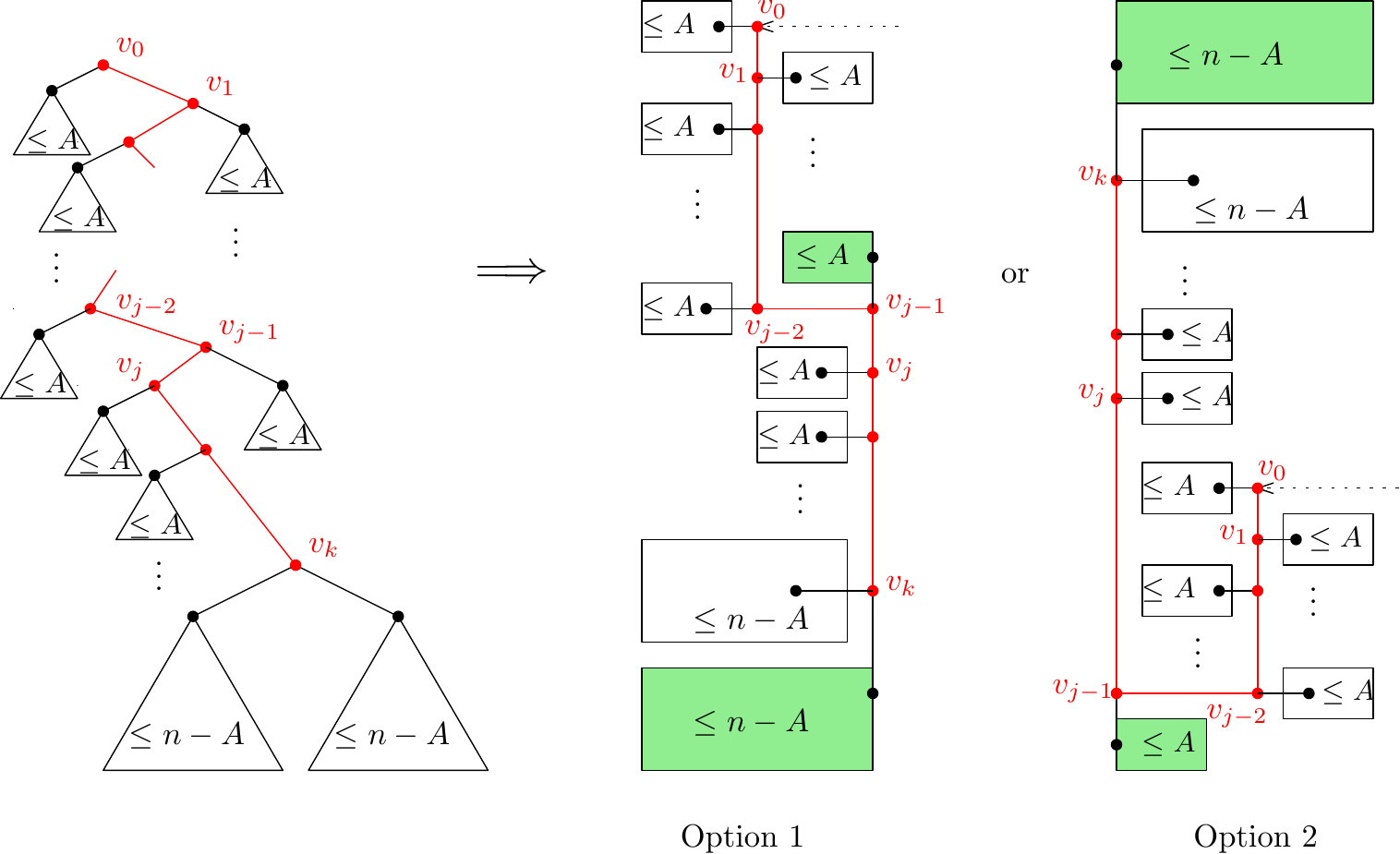}{\PAPER{1}\LIPICS{0.9}}{Algorithm for orthogonal order-preserving drawings.  Option~1 is for the case when $v_{j-1}$ is a right 
child (as in the tree depicted).  Option~2 is for the
case when $v_{j-1}$ is a left child instead.}

An issue arises from the shaded parts in the figure, i.e., the drawings of
the subtree at the sibling of $v_j$ and at the right child of $v_k$.
In these drawings, the root of such a subtree needs to be
reachable vertically from the left or right side of the bounding box (rather than horizontally).  Fortunately, such a drawing can be obtained recursively as shown
in Figure~\ref{fig_aux}(a).

\fig{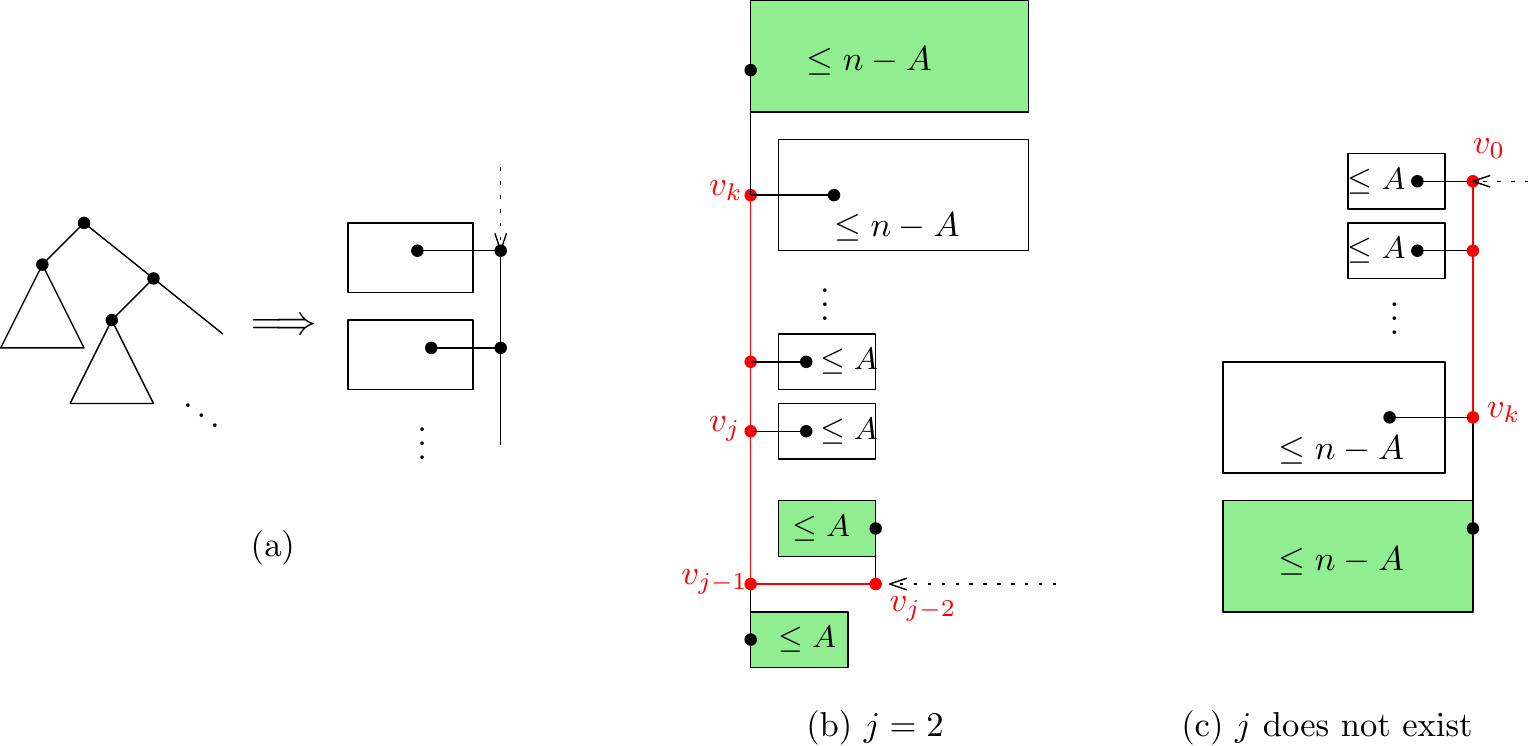}{\PAPER{1}\LIPICS{0.9}}{(a) Making the root reachable vertically from the right side of the bounding box. (b) The special case $j=2$. (c) The special case when $j$ does not exist (i.e., $v_0v_1\cdots v_k$ is a rightward path).}

The special cases when $j=2$ or when $j$ does not exist are described
in Figure~\ref{fig_aux}(b,c).  (The special case $j=1$ cannot occur, since $v_1$ is assumed to be a right child.)

The case when $v_1$ is a left child can be handled by
flipping all the drawings upside-down and swapping ``left'' with ``right''.

The overall width satisfies the recurrence
\[ W(n)\ \le\ \max\{2\,W(A),\,W(n-A)\} + O(1).
\]
Iterating on the second term over a common $A$ gives
$W(n)\:\le\: 2\,W(A) + O(n/A).$
Setting $A=n/2^{\sqrt{2\log n}}$ gives
\[ W(n)\ \le\ 2\,W(n/2^{\sqrt{2\log n}})+O(2^{\sqrt{2\log n}}),
\]
which solves to $W(n)=O(2^{\sqrt{2\log n}}\sqrt{\log n})$,
as shown in~\cite{SODA99}.
The height of the drawing is trivially bounded by $n$ (since
each row should contain at least one node).

\begin{theorem}
Every binary tree of size $n$ has a straight-line orthogonal order-preserving
drawing with area $O(n2^{\sqrt{2\log n}}\sqrt{\log n})$.
\end{theorem}

\noindent {\em Remarks.} 
It is straightforward to implement the algorithm to run in linear time. 

The original $n2^{O(\sqrt{\log n})}$ algorithm by Chan~\cite{SODA99}
for straight-line, strictly upward, order-preserving  drawings
of binary trees was subsequently surpassed by the $O(n\log n)$
algorithm by Garg and Rusu~\cite{GarRus03}.
However, we do not see how to adapt Garg and Rusu's approach to help
improve our result here.

\small
\bibliographystyle{abbrvurl}
\bibliography{more_trees}

\begin{thebibliography}{10}

\bibitem{AgHaVa}
P.~K. Agarwal, S.~Har{-}Peled, and K.~R. Varadarajan.
\newblock Approximating extent measures of points.
\newblock {\em J. {ACM}}, 51(4):606--635, 2004.
\newblock \href {http://dx.doi.org/10.1145/1008731.1008736}
  {\path{doi:10.1145/1008731.1008736}}.

\bibitem{Bac}
C.~Bachmaier, F.~Brandenburg, W.~Brunner, A.~Hofmeier, M.~Matzeder, and
  T.~Unfried.
\newblock Tree drawings on the hexagonal grid.
\newblock In {\em Proc. 16th International Symposium on Graph Drawing (GD)},
  pages 372--383, 2008.
\newblock \href {http://dx.doi.org/10.1007/978-3-642-00219-9_36}
  {\path{doi:10.1007/978-3-642-00219-9_36}}.

\bibitem{BarHar}
G.~Barequet and S.~Har{-}Peled.
\newblock Efficiently approximating the minimum-volume bounding box of a point
  set in three dimensions.
\newblock {\em J. Algorithms}, 38(1):91--109, 2001.
\newblock \href {http://dx.doi.org/10.1006/jagm.2000.1127}
  {\path{doi:10.1006/jagm.2000.1127}}.

\bibitem{BieJGAA17}
T.~Biedl.
\newblock Ideal drawings of rooted trees with approximately optimal width.
\newblock {\em J. Graph Algorithms Appl.}, 21:631--648, 2017.
\newblock \href {http://dx.doi.org/10.7155/jgaa.00432}
  {\path{doi:10.7155/jgaa.00432}}.

\bibitem{BieCCCG17}
T.~Biedl.
\newblock Upward order-preserving 8-grid-drawings of binary trees.
\newblock In {\em Proc. 29th Canadian Conference on Computational Geometry
  (CCCG)}, pages 232--237, 2017.
\newblock \url{http://2017.cccg.ca/proceedings/Session6B-paper4.pdf}.

\bibitem{SODA99}
T.~M. Chan.
\newblock A near-linear area bound for drawing binary trees.
\newblock {\em Algorithmica}, 34(1):1--13, 2002.
\newblock \href {http://dx.doi.org/10.1007/s00453-002-0937-x}
  {\path{doi:10.1007/s00453-002-0937-x}}.

\bibitem{GD96}
T.~M. Chan, M.~T. Goodrich, S.~R. Kosaraju, and R.~Tamassia.
\newblock Optimizing area and aspect ration in straight-line orthogonal tree
  drawings.
\newblock {\em Comput. Geom.}, 23(2):153--162, 2002.
\newblock \href {http://dx.doi.org/10.1016/S0925-7721(01)00066-9}
  {\path{doi:10.1016/S0925-7721(01)00066-9}}.

\bibitem{Cre}
P.~Crescenzi, G.~{Di Battista}, and A.~Piperno.
\newblock A note on optimal area algorithms for upward drawings of binary
  trees.
\newblock {\em Comput. Geom.}, 2:187--200, 1992.
\newblock \href {http://dx.doi.org/10.1016/0925-7721(92)90021-J}
  {\path{doi:10.1016/0925-7721(92)90021-J}}.

\bibitem{CrePen}
P.~Crescenzi and P.~Penna.
\newblock Strictly-upward drawings of ordered search trees.
\newblock {\em Theor. Comput. Sci.}, 203(1):51--67, 1998.
\newblock \href {http://dx.doi.org/10.1016/S0304-3975(97)00287-9}
  {\path{doi:10.1016/S0304-3975(97)00287-9}}.

\bibitem{FrPaPo}
H.~de~Fraysseix, J.~Pach, and R.~Pollack.
\newblock How to draw a planar graph on a grid.
\newblock {\em Combinatorica}, 10(1):41--51, 1990.
\newblock \href {http://dx.doi.org/10.1007/BF02122694}
  {\path{doi:10.1007/BF02122694}}.

\bibitem{DiBOOK}
G.~{Di Battista}, P.~Eades, R.~Tamassia, and I.~G. Tollis.
\newblock {\em Graph Drawing}.
\newblock Prentice Hall, 1999.

\bibitem{DiFraSURV}
G.~{Di Battista} and F.~Frati.
\newblock A survey on small-area planar graph drawing.
\newblock {\em CoRR}, abs/1410.1006, 2014.
\newblock \href {http://arxiv.org/abs/1410.1006} {\path{arXiv:1410.1006}}.

\bibitem{Fra07}
F.~Frati.
\newblock Straight-line orthogonal drawings of binary and ternary trees.
\newblock In {\em Proc. 15th International Symposium on Graph Drawing (GD)},
  pages 76--87, 2007.
\newblock \href {http://dx.doi.org/10.1007/978-3-540-77537-9_11}
  {\path{doi:10.1007/978-3-540-77537-9_11}}.

\bibitem{FraSODA17}
F.~Frati, M.~Patrignani, and V.~Roselli.
\newblock {LR}-drawings of ordered rooted binary trees and near-linear area
  drawings of outerplanar graphs.
\newblock In {\em Proc. 28th {ACM-SIAM} Symposium on Discrete Algorithms
  (SODA)}, pages 1980--1999, 2017.
\newblock \href {http://dx.doi.org/10.1137/1.9781611974782.129}
  {\path{doi:10.1137/1.9781611974782.129}}.

\bibitem{GarSoCG93}
A.~Garg, M.~T. Goodrich, and R.~Tamassia.
\newblock Planar upward tree drawings with optimal area.
\newblock {\em Int. J. Comput. Geometry Appl.}, 6(3):333--356, 1996.
\newblock Preliminary version in SoCG'93.
\newblock \href {http://dx.doi.org/10.1142/S0218195996000228}
  {\path{doi:10.1142/S0218195996000228}}.

\bibitem{GarRus03}
A.~Garg and A.~Rusu.
\newblock Area-efficient order-preserving planar straight-line drawings of
  ordered trees.
\newblock {\em Int. J. Comput. Geometry Appl.}, 13(6):487--505, 2003.
\newblock \href {http://dx.doi.org/10.1142/S021819590300130X}
  {\path{doi:10.1142/S021819590300130X}}.

\bibitem{GarRusICCSA}
A.~Garg and A.~Rusu.
\newblock Straight-line drawings of general trees with linear area and
  arbitrary aspect ratio.
\newblock In {\em Proc. 3rd International Conference on Computational Science
  and Its Applications (ICCSA), Part {III}}, pages 876--885, 2003.
\newblock \href {http://dx.doi.org/10.1007/3-540-44842-X_89}
  {\path{doi:10.1007/3-540-44842-X_89}}.

\bibitem{GarRus04}
A.~Garg and A.~Rusu.
\newblock Straight-line drawings of binary trees with linear area and arbitrary
  aspect ratio.
\newblock {\em J. Graph Algorithms Appl.}, 8(2):135--160, 2004.
\newblock \url{http://jgaa.info/accepted/2004/GargRusu2004.8.2.pdf}.

\bibitem{LeeTHESIS}
S.~Lee.
\newblock Upward octagonal drawings of ternary trees.
\newblock Master's thesis, University of Waterloo, 2016.
\newblock (Supervised by T. Biedl and T. M. Chan.)
  \url{https://uwspace.uwaterloo.ca/handle/10012/10832}.

\bibitem{Lei}
C.~E. Leiserson.
\newblock Area-efficient graph layouts (for {VLSI)}.
\newblock In {\em Proc. 21st IEEE Symposium on Foundations of Computer Science
  (FOCS)}, pages 270--281, 1980.
\newblock \href {http://dx.doi.org/10.1109/SFCS.1980.13}
  {\path{doi:10.1109/SFCS.1980.13}}.

\bibitem{Mat}
J.~Matou{\v{s}}ek.
\newblock Range searching with efficient hiearchical cutting.
\newblock {\em Discrete {\&} Computational Geometry}, 10:157--182, 1993.
\newblock \href {http://dx.doi.org/10.1007/BF02573972}
  {\path{doi:10.1007/BF02573972}}.

\bibitem{ReiTil}
E.~M. Reingold and J.~S. Tilford.
\newblock Tidier drawings of trees.
\newblock {\em {IEEE} Trans. Software Eng.}, 7(2):223--228, 1981.
\newblock \href {http://dx.doi.org/10.1109/TSE.1981.234519}
  {\path{doi:10.1109/TSE.1981.234519}}.

\bibitem{Sch}
W.~Schnyder.
\newblock Embedding planar graphs on the grid.
\newblock In {\em Proc. 1st {ACM-SIAM} Symposium on Discrete Algorithms
  (SODA)}, pages 138--148, 1990.
\newblock \url{http://dl.acm.org/citation.cfm?id=320176.320191}.

\bibitem{ShiTHESIS}
Y.~Shiloach.
\newblock {\em Linear and Planar Arrangement of Graphs}.
\newblock PhD thesis, Weizmann Institute of Science, 1976.
\newblock
  \url{https://lib-phds1.weizmann.ac.il/Dissertations/shiloach_yossi.pdf}.

\bibitem{ShKiCh}
C.~Shin, S.~K. Kim, and K.~Chwa.
\newblock Area-efficient algorithms for straight-line tree drawings.
\newblock {\em Comput. Geom.}, 15(4):175--202, 2000.
\newblock Preliminary version in COCOON'96.
\newblock \href {http://dx.doi.org/10.1016/S0925-7721(99)00053-X}
  {\path{doi:10.1016/S0925-7721(99)00053-X}}.

\bibitem{ShiIPL}
C.~Shin, S.~K. Kim, S.~Kim, and K.~Chwa.
\newblock Algorithms for drawing binary trees in the plane.
\newblock {\em Inf. Process. Lett.}, 66(3):133--139, 1998.
\newblock \href {http://dx.doi.org/10.1016/S0020-0190(98)00049-0}
  {\path{doi:10.1016/S0020-0190(98)00049-0}}.

\bibitem{Tre}
L.~Trevisan.
\newblock A note on minimum-area upward drawing of complete and fibonacci
  trees.
\newblock {\em Inf. Process. Lett.}, 57(5):231--236, 1996.
\newblock \href {http://dx.doi.org/10.1016/0020-0190(96)81422-0}
  {\path{doi:10.1016/0020-0190(96)81422-0}}.

\bibitem{Val}
L.~G. Valiant.
\newblock Universality considerations in {VLSI} circuits.
\newblock {\em {IEEE} Trans. Computers}, 30(2):135--140, 1981.
\newblock \href {http://dx.doi.org/10.1109/TC.1981.6312176}
  {\path{doi:10.1109/TC.1981.6312176}}.

\bibitem{VazBOOK}
V.~V. Vazirani.
\newblock {\em Approximation Algorithms}.
\newblock Springer, 2001.
\newblock \url{https://www.cc.gatech.edu/fac/Vijay.Vazirani/book.pdf}.

\end{thebibliography}
}

\end{document}